\documentclass[journal]{IEEEtran}
\usepackage{ifpdf}

\newtheorem{proposition}{Proposition}

\ifCLASSOPTIONcompsoc
 \usepackage[nocompress]{cite}
\else
 \usepackage{cite}
\fi
\ifCLASSINFOpdf
 \usepackage[pdftex]{graphicx}
 \DeclareGraphicsExtensions{.pdf,.jpeg,.png}
\else
\usepackage[dvips]{graphicx}
 \DeclareGraphicsExtensions{.eps}
\fi
\usepackage[cmex10]{amsmath}

\usepackage{array}

\usepackage{mdwtab}

\begin{document}
%
\title{Channel Independent Cryptographic Key Distribution}
%
%
%
%

\author{Benjamin.~T.~H.~Varcoe\protect\\
\thanks{School of Physics and Astronomy, University of Leeds\protect\\
E-mail: b.varcoe@leeds.ac.uk
}
\thanks{Manuscript received .............}}

%


\markboth{IEEE TRANSACTIONS ON INFORMATION THEORY,~Vol.~ , No.~ , ~20??}%
{Varcoe: Channel Independent Cryptography}

\maketitle

{%
\begin{abstract}
This paper presents a method of cryptographic key distribution using an `artificially' noisy channel.
This is an important development because, while it is known that a noisy channel can be used to generate unconditional secrecy, there are many circumstances in which it is not possible to have a noisy information exchange, such as in error corrected communication stacks.
It is shown that two legitimate parties can simulate a noisy channel by adding local noise onto the communication and that the simulated channel has a secrecy capacity even if the underlying channel does not. 
A derivation of the secrecy conditions is presented along with numerical simulations of the channel function to show that key exchange is feasible. 
\end{abstract}

}


\IEEEdisplaynotcompsoctitleabstractindextext

%
\IEEEpeerreviewmaketitle

\section{Introduction}

\IEEEPARstart{A} fundamental aspect of cryptography is the ability to exchange information in perfect secrecy.
However, almost all current forms of encryption rely on computational rather than information theoretic security.
Computationally secure coding uses encryption techniques for which the known hacking algorithms require lengthy calculations and they therefore assume that the eavesdropper has a limited computational capability. 
Information theoretic security, on the other hand, makes no assumptions about an eavesdroppers capability. 
This is an important distinction because the computational difficulty can change with advances in algorithms and increased speed of calculation, leading to significant cryptographic breaks. 

One solution to producing a perfectly encrypted message is the One Time Pad where the cryptographer has an infinite supply of symmetric key that is consumed as a resource in encrypting information. 
However, this creates the obvious problem of creating and securely distributing large amounts of symmetric key material.
To address this problem several methods have been proposed for creating and agreeing on a perfectly secure key of arbitrary length under information theoretic security\cite{Grosshans2003,Gisin2002,Wyner1975,Csiszar1978,Maurer1997,Bennett1984,Maurer1993,Maurer1999}.

In particular, it has been shown \cite{Wyner1975,Csiszar1978,Maurer1993,Ahlswede1993} that channel noise can be exploited to create and distribute symmetric key material in near perfect secrecy.
The essential element is the use of feedback \cite{Maurer1993,Maurer1997,Bennett1988,Bennett1995,Maurer1999,Ahlswede1993}, which exploits uncertainties in Eve's (the eavesdropper's) information, via a specially prepared sequence of messages, to generate random bit sequences known to the  communicating pair, Alice and Bob, and unknown to Eve.
Hence, the  secrecy of the channel is given by the combination of  information transmitted over both the noisy and noise-free channels.


The noisy channel model has been studied for several decades now (see for example \cite{Wyner1975,Csiszar1978,Bennett1988,Bennett1984,Maurer1993,Maurer1999}) and it has been well established that when the mutual information between two distributions A and B exceeds the mutual information between A and E (that is $I(A;B)>I(A;E)$) it is possible to build a secret key from distributions A and B excluding E.
Maurer \cite{Maurer1993} recognized that this condition was too strong and there are circumstances when $I(A;B)<I(A;E)$ and/or $I(A;B)<I(B;E)$ but secret key generation is still possible. 
In particular, feedback between Alice and Bob permits the generation of a sequence of messages $M^t=[M_1,M_2,\dots,M_t]$ which can be used to generate $\hat{A}=M^t A$, $\hat{B}=M^t B$ and $\hat{E}=M^t E$ such that
\[I(\hat{A};\hat{B})>I(\hat{A};\hat{E}).\] 

 
The scenario works as follows, Alice and Bob first exchange symbols over the noisy channel. 
They then use feedback signals exchanged over an error-corrected channel to arrive at a secret key. 
Information Alice, Bob and Eve obtain about the secret key therefore comes from correlations in the signals and the joint information gained from the subsequent feedback messages.

The conditions of feedback are therefore an important aspect of distilling the secret key. 
The feedback must be constructed to allow Alice and Bob to use correlations in their data to exclude Eve. 
One way for them to do this is for Alice and Bob to reduce the errors in a subsequent communication channel. 
There are several codings that could be used for this purpose\cite{VanDijk1998}.
As Eve has errors in her knowledge of both Alice's and Bob's bit sequences, the coding only gives her partial information and she therefore receives the signals with errors cascaded over either Alice's or Bob's channel (depending on the nature of the reconciliation). 
It has been shown \cite{VanDijk1998,Maurer1999,Maurer1993} that this gives Alice and Bob an advantage over Eve.
This procedure can be cascaded to align their information while maintaining a well defined upper limit on Eve's information.
Moreover, knowing the upper limit on Eve's information, Alice and Bob can subsequently use privacy amplification \cite{Bennett1988,Maurer1997} to reduce Eve's information to an arbitrarily small level. 

The essential element of secret key generation is that the  noise channel is capable of generating a sufficient and predictable level of uncertainty in the information received by Eve.
Quantum key distribution\cite{Gisin2002,Grosshans2003} is one way of ensuring a guaranteed minimum level of channel noise and it comes with the added benefit that the security is protected by the physical limitations of quantum mechanical measurement.
However, as noted above, any noisy communications channel where Alice, Bob and Eve receive noisy copies of a signal, is also theoretically secure when it is assumed that each party receives uncorrelated noise. 


Hence, an implementation of any cryptographic scheme will be strongly dependent on the specific properties of the channel and all of the apparatus surrounding it, because changes to the detectors, source or even medium (e.g. free space or waveguide) can change the noise and hence the secrecy conditions.
Secrecy can even be removed entirely if the noise falls below a given threshold. 
Continuous variable quantum communication, for example, typically includes essential device calibration \cite{Fossier2009}.
For this reason it is interesting to consider a method that would generate unconditional channel noise and therefore secrecy in key generation.

The current paper presents a method of applying {\it unconditional} noise to generate secrecy in a channel with no {\it a priori} secrecy capacity. 
Specifically it is found that if the signal is degraded locally (independent of transmission or detection noise) and by {\it both} Alice and Bob independently, the channel gains the ability to create a secret symmetric key even if it otherwise has no prior secrecy capacity.

\section{Secret Keys from Channel Noise}

It has been established that a noisy channel where, in particular, Eve and Bob both have reception errors is capable of secret communication while a perfect, noise-free communication channel has zero secrecy capacity\cite{Maurer1993,Maurer1999, Maurer2007,Wyner1975,Csiszar1978,Bennett1988,Bennett1984}. 
However there are many channels in where noisy communication is not possible, or where an unconditionally noisy channel is not guaranteed. 
In such channels we propose that Alice and Bob use a local noise source to create an artificially noisy channel, simply by adding local noise {\it after} the transmission (Fig.\ref{ChannelPicture}).
Although Eve can receive a noiseless copy of the original data, it will now be shown that by altering their channel Alice and Bob have effectively simulated a noisy channel and can distill a correlated data set unknown to Eve.
Alice and Bob have effectively generated a non-zero secrecy capacity at the expense of reducing the overall channel capacity.

In the arguments that follow it is assumed that Eve can detect all transmissions from Alice (and Bob) with no error and that she receives a signal that is identical to the receiver at all times.
This represents the best possible case for Eve.


\begin{figure}[h,t,b]
\begin{center}
\begin{picture}(100,100)(-60,-20)
\put(-70,0){\circle{8}}
 
 \put(-90,20){\framebox(40,20){$N^A_i$}}
 \put(-70,-4){\line(0,1){23}}
 \put(-70,15){\vector(0,-1){10}}
 
 \put(-85,55){Alice}
 \put(-110,-15){\dashbox{1}(80,85)}

 \put(-105,-10){\framebox(20,20){$X_i$}}
 \put(-55,-10){\framebox(20,20){$R_i$}}
 \put(-35,0){\vector(1,0){135}}
 \put(-55,0){\vector(-1,0){30}}

\put(80,0){\circle{8}}
  \put(75,55){Bob}
  \put(60,20){\framebox(40,20){$N^B_i$}}
  \put(80,-4){\line(0,1){23}}
 \put(80,15){\vector(0,-1){10}}
 \put(-25,15){\dashbox{1}(75,55)}

  \put(100,-10){\framebox(20,20){$Y_i$}}

  \put(-5,0){\vector(0,1){45}}
  \put(-10,26){\large$\oplus$}
 \put(10,55){Eve}
  \put(-15,45){\framebox(20,20){$Z_i$}}
  \put(5,20){\framebox(40,20){$N^E_i$}}
  \put(-8,29){\line(1,0){13}}
  \put(5,29){\vector(-1,0){7}}
 \put(55,-15){\dashbox{1}(70,85)}

\end{picture}
\end{center}
\caption{This figure presents the communication scheme, where Alice initially transmits R to Bob (and Eve) over a public line. Alice and Bob (and Eve) then add noise to the random variable using local noise sources, $N^A$, $N^B$ and $N^E$. This results in the variables $X$, $Y$ and $Z$ for Alice, Bob and Eve respectively, which will be the starting point for further communication.} 
\label{ChannelPicture}
\end{figure}
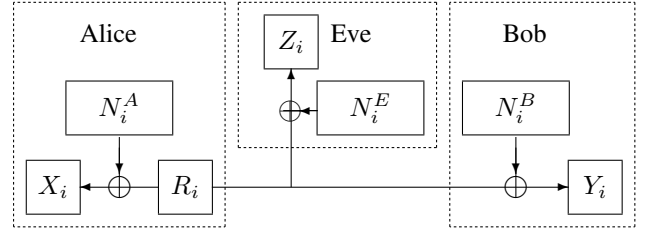

In this model (without loss of generality) Alice transmits a random number $R$ to Bob over an otherwise noiseless channel. 
Alice, Bob and Eve then use a local true random number generator to generate the random number sets $N^A$, $N^B$ and $N^E$ respectively, such that Eve cannot independently deduce the values $N^A$ or $N^B$ from her knowledge of $N^E$ and $R$.
Alice and Bob and Eve use these to generate the variables 
\begin{equation}X_i=R_i\oplus N^A_i,\label{noise1}\end{equation} 
\begin{equation}Y_i=R_i\oplus N^B_i\label{noise2}\end{equation} 
and 
\begin{equation}Z_i=R_i\oplus N^E_i\label{noise3}\end{equation}
respectively.
$N^A_i$, $N^B_i$ and $N^E_i$ are random biased bits with probabilities $P(N^A_i=0)=1-\alpha$, $P(N^A_i=1)=\alpha$, $P(N^B_i=0)=1-\beta$, $P(N^B_i=1)=\beta$ (where $0<\alpha,\beta<1$ ), $P(N^E_i=0)=1-\gamma$ and $P(N^E_i=1)=\gamma$ (where $0\leq\gamma\leq1$).  
The range of $\gamma$ expresses the fact that the value $N^E_i$ could contain any level of noise ranging from complete uncertainty to being an error free reception of $R_i$.

This effectively creates new binary symmetric channels between Alice and Bob and Alice and Eve shown in Fig. \ref{ErrorPicture}.
Where the bit flip probabilities $\alpha$, $\beta$ and $\gamma$ are related to the binary symmetric error rates $\epsilon$ and $\delta$ via $\epsilon =\alpha+\beta-2\alpha\beta$ and $\delta=\alpha+\gamma-2\alpha\gamma$. 
Importantly both $\delta>0$ and $\epsilon>0$ even if $\gamma=0$ as long as both $\alpha>0$ and $\beta>0$.

It is worth noting that Eve might tamper with the initial transmission of R, which effectively amounts to creating a new variable $T_i$, where $P(T_i=0)=1-\tau$, $P(T_i=1)=\tau$. 
This would lead to Bob receiving a new variable $Y'_i=Y_i \oplus T_i$. 
This has the same effect as transmission noise and may make key agreement between Alice and Bob more difficult, but it is not impossible as long as $I(Y;Y')\neq 0\neq I(X;Y') $ \cite{Maurer1993}. 
Tampering is considered in more detail below.

\begin{figure}[t,h,b]
\begin{center}
\begin{picture}(80,100)(-60,-45)
\put(-45,-20){\small $1-\delta$}
\put(-45,20){\small $1-\delta$}
\put(-45,8){\small $\delta$}
\put(-45,-8){\small $\delta$}

\put(-115,-25){\dashbox{1}(30,50)}
\put(-108,30){Eve}
\put(-103,10){0}
\put(-103,-15){1}
\put(-102,-40){Z}

\put(-15,-25){\dashbox{1}(30,50)}
\put(-8,30){Alice}
\put(-3,10){0}
\put(-3,-15){1}
\put(-2,-40){X}

\put(10,17){\vector(1,0){80}}
\put(10,15){\vector(3,-1){80}}

\put(10,-12){\vector(1,0){80}}
\put(10,-10){\vector(3,1){80}}

\put(85,-25){\dashbox{1}(30,50)}
\put(92,30){Bob}
\put(97,10){0}
\put(97,-15){1}
\put(98,-40){Y}

\put(-10,17){\vector(-1,0){80}}
\put(-10,15){\vector(-3,-1){80}}

\put(-10,-12){\vector(-1,0){80}}
\put(-10,-10){\vector(-3,1){80}}

\put(25,-20){\small $1-\epsilon$}
\put(25,20){\small $1-\epsilon$}
\put(37,9){\small $\epsilon$}
\put(37,-9){\small $\epsilon$}

\end{picture}
\end{center}
\caption{The transmission line model presented in figure \ref{ChannelPicture} in which Alice transmits a signal to Bob over an error free channel with local degradation of the data, can be rewritten as a series of binary symmetric channels linking Alice and Bob and Alice and Eve, where $\epsilon =\alpha+\beta-2\alpha\beta$ is the convolution of local noise between Alice and Bob and $\delta=\alpha+\gamma-2\alpha\gamma$ is the convolution of local noise between Alice and Eve, with $\alpha$, $\beta$ and $\gamma$ as local bit flip probabilities (defined in the text). }
\label{ErrorPicture}
\end{figure}
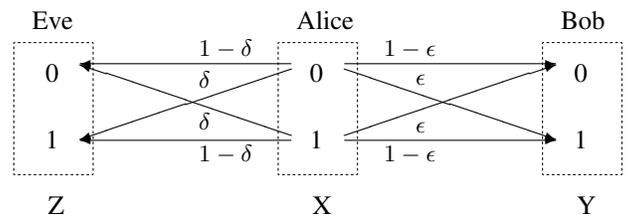


It is therefore proposed that for $\alpha>0$, $\beta>0$ and $\gamma\ge 0$ a channel has been created with a capacity for secrecy, irrespective of the properties of the underlying channel.
This is a new condition introduced in this paper that does not match the conditions established previously \cite{Maurer1999,Maurer1993,Ahlswede1993}.
It is therefore not immediately clear that a protocol can even be constructed with the capacity for secrecy.
The secrecy proof is re-visited in the following proposition. 

\begin{proposition}
For the binary distribution $X$, $Y$ and $Z$ given by the conditions in equations \ref{noise1}, \ref{noise2} and \ref{noise3} and assuming $\alpha=\beta$, the proposition is that for any $\gamma\geq 0$ a coding sequence exists such that variables $\hat{X}$, $\hat{Y}$ and $\hat{Z}$ can be derived from $X$, $Y$, and $Z$ such that the channel has a positive secrecy capacity.
\end{proposition}

\begin{proof}
Alice and Bob open a virtual channel where Alice encodes a random number, $C$, using a simple encoding of sending an $N$-fold repeated message\cite{Maurer1993} encrypted using N-bits from $X$ via $\left[X_i \oplus C,X_{i+1} \oplus C,\dots, X_{i+N-1}\oplus C\right]$.
Bob calculates the sequence $C'$ using the sequence $Y$ via $\left[X_i \oplus C\oplus Y_i,X_{i+1} \oplus C\oplus Y_{i+1},\dots, X_{i+N-1} \oplus C\oplus Y_{i+N-1}\right]$ and accepts the value of $C'$ only when he has obtained the vectors $\left[0,0,\dots ,0\right]$ or $\left[1,1,\dots ,1\right]$. 
This encoding therefore creates a new binary symmetric channel between Alice and Bob with error probability, $\epsilon_{N}$, given by
\begin{equation}
\epsilon_{N}=\frac{P_{\rm Error}}{P_{\rm Correct}+P_{\rm Error}}=\frac{\epsilon^N}{(1-\epsilon)^N+\epsilon^N},
\label{beta}
\end{equation}
where $P_{Error}$ is the probability that the bit calculated by Bob is not equal to the bit encoded by Alice and $P_{Correct}$ is the probability that the two are equal. 
Eve can also calculate `$C$' using her variables $\left[Z_i,Z_{i+1},\dots,Z_{i+N-1}\right]$, however, it is assumed that Eve takes a more sophisticated approach than Bob and calculates the value of `$C$' based on the most likely value.
Errors in Eve's calculation are therefore determined by the number of sequences with $N/2$ errors or more, given that the sequence has been accepted by Bob. 
Eve's error in the $N$-fold virtual channel is $\delta_{N}$, assuming Eve adds no noise (i.e. $\gamma=0$).
Hence, the best match that Eve can achieve to Bob's data is given by the binomial probability\cite{Maurer1993,Maurer1999},
\begin{equation}
\delta_{N}=\frac{1}{P_{Total}}\sum_{w=N/2}^N\left(
\begin{array}{cc}
N\\
w
\end{array}
\right)
\left(p_{00}^{N-w}p_{01}^w+p_{10}^{N-w}p_{11}^w\right)
\label{gamma}
\end{equation}
where, $P_{Total}$ is the total number of states accepted by Bob, $p_{nm}$ is the probability that when the result of Alice's calculation ($R+Noise$) is $0$, Bob's result is $n$ and Eve's result is $m$.
In this specific case their values are given by,\\ 
\[P_{Total}= P_{Correct}+P_{Error}= (1-\epsilon)^N+\epsilon^N,\]
\[p_{00}= (1-\alpha)^2,\]
\[p_{01}= \alpha^2, \]
and
\[p_{10}=p_{11}=\alpha(1-\alpha).\]

Considering only the special case that Bob has accepted a sequence and the sequence decoded by Eve has equal numbers of 1s and 0s (that is, Eve has a 50\% transmission error), we are left with, 
\[
\delta_N > \frac{1}{P_{Total}}\left( 
\begin{array}{cc}
N\\
N/2
\end{array}
\right)\left(2(\alpha-\alpha^2)^{N}\right).\]
At this point a choice of coding is made, deciding on the specific value of $N=2$. 
Under these conditions the errors in Eve's calculation are given by,
\[\delta_{N=2} > \frac{4(\alpha-\alpha^2)^{2}}{P_{Total}}. \]
However equation \ref{beta} shows that when $N=2$ and $\alpha=\beta$, Bob's error is also given by
\begin{equation} \epsilon_{N=2}=\frac{\epsilon^2}{P_{Total}}=\frac{(2\alpha-2\alpha^2)^2}{P_{Total}}. \label{BobError}\end{equation}
Hence, $\delta_{N=2}$ is {\it always} strictly greater than $\epsilon_{N=2}$. 
It follows that $I(\hat{X};\hat{Y})>I(\hat{Y};\hat{Z})$ and therefore $S(\hat{X};\hat{Y}||\hat{Z})>0$ \cite{Maurer1999}.
This is sufficient even if it is only true for the special case of $N=2$.
\end{proof}


This proposition is useful because it also provides a protocol with which a secret key can be established. 
Alice can repeat the exchange with same transmission steps (and different random variables $C$), knowing that Eve's ability to calculate the value will be worse than Bob's because Eve started with an imperfect code $\hat{Z}$.
This can be repeated until Alice and Bob have $n$-bit sets $S_n$ and $S'_n$, where $Prob(S\neq S'_n)<\kappa $. 
Likewise Eve will have $S''_n$ for which $Prob(S=S''_n)<\lambda$, for some choice of $\kappa$ and $\lambda$.

For sufficiently small $\kappa$ we have 
\[I(S_n;S'_n)-I(S_n;S''_n)\approx  1-I(S_n;S''_n)\leq k,\]
indicating that Eve's knowledge of $S_n$ is upper bounded by $k$. 
Therefore, privacy amplification\cite{Bennett1995} can be used then be used to further reduce Eve's information using a hash function $F_{n,k}$ which performs a mapping from Alice and Bob's $n$ bit sequence to a sequence of length $n(1-k)$ bits.
This creates new $n(1-k)$ length sets $F_{n,k}(S'')$ and $F_{n,k}(S')=F_{n,k}(S)$ such that $I(F_{n,k}(S'');F_{n,k}(S))\approx0$. 
If the hash function $F_{n,k}$ is decided after the protocol has been completed, Eve cannot influence the exchanges in advance in order to create a favorable situation. 




Figure \ref{ErrorRate} presents the results of a simulation of random data exchanged between Alice and Bob, followed by adding local noise (the first two steps in the protocol).
The simulation modeled a real exchange, where bit strings and added noise were generated using a random number generator and the signal was sent virtually between analysis programs. 
In this case the random number generator fails a universal statistical test \cite{Maurer1991} so this particular exchange has limited capacity for generating actual secrecy, however, it does allow us to investigate the statistical effects of added noise.
The solid line is an evaluation of equations \ref{beta} and \ref{gamma} and the open and closed circles are the results of the simulation for Bob and Eve respectively.
This simulation shows that for all values of noise $\alpha>0$ to $\alpha<1/2$ (with $\alpha=\beta$) the noise in Eve's channel exceeds the noise in Bob's channel following the first exchange.

\begin{figure}[!t]
\centering
\begin{picture}(100,180)(10,0)
6\put(-60,0){\includegraphics[width=3.2in]{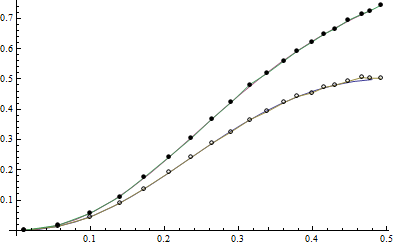}}
\put(-60,150){\it $N=2$ Channel Error}
\put(45,95){$\epsilon_{N=2}$}
\put(70,95){\vector(1,-1){20}}
\put(45,118){$\delta_{N=2}$}
\put(70,118){\vector(1,-1){20}}
\put(170,-5){$\alpha$}
\end{picture}
\caption{Theory vs simulation evaluations of the binary symmetric channel errors $\gamma_N$ and $\epsilon_N$ for the virtual channel for N=2 as a function of the error rate $\alpha$. The solid line is from equations \ref{beta} and \ref{gamma} and the points are the results of the simulations. For this graph, $\alpha=\beta$ and 100,000 samples were taken. It shows that the error rate for Eve always exceeds that of Bob for all choices of $\alpha$. }
\label{ErrorRate}
\end{figure}
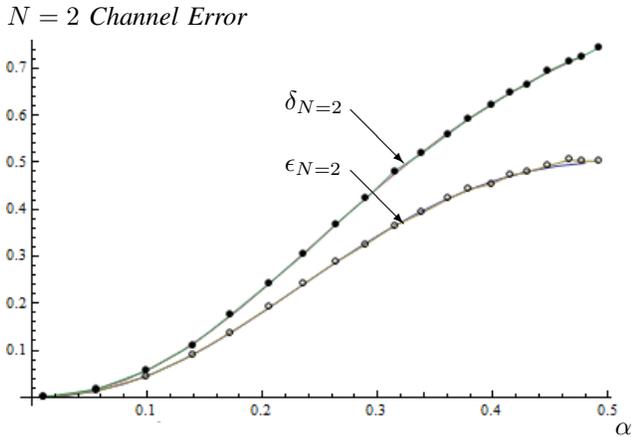

\section{Eavesdropping}
It is also worth considering the type of eavesdropping attacks that might be expected and how they are thwarted.
For this purpose, the system can be broken down into two elements; the private channel in which the key material is exchanged and the public channel in which the two way communication used to distill the secret key takes place. 

The private channel is always used for the exchange of random variables which will be unknown to Bob in advance of the exchange. 
Hence it is possible that Eve could tamper or otherwise interfere with these transmissions by removing, altering or inserting symbols without Bob knowing. 

It is also assumed that Eve is only able to read but not modify messages in the authenticated public channel. 
Authentication in this case meaning that Alice and Bob are authenticated users rather than referring to the security of the connection itself. 

Eve has two possible attacks. Firstly she can look for correlations in the data exchanged and try to use this to gain insight into the final key (more than the insight she gains by simply following Alice and Bob's protocol). 
Secondly, she can manipulate the symbols exchanged; for example, by performing selected bit flips in some of the data exchanged. 

The first attack carries no additional information because the key is not sourced from the initial data but rather from the elements $C$ used in subsequent exchanges. 
Moreover the noise, $N^A$ and $N^B$ and the virtual channel value $C$ are chosen such that $I(X;C)=I(Y;C)=I(Z;C)=0$, hence no information about $C$ can be gained from observations of the initial code $R$.

The second attack is somewhat more important because it can go unnoticed especially as Bob is already expecting unknown random data.
However, as the statistics cannot be altered from the original and the number of bits must still be the same, Bob can perform local checks to ensure that the data is still random.
Therefore the attack may be quite subtle.
For example, it could be assumed that Eve's goal in manipulating the data is not to reveal the secret key itself but rather to reveal a few additional bits of each exchange that will be unaccounted for in the privacy amplification. 
Eve may use this to gain a cryptographic break over the final key (Eve may, for example, wish to reduce the search space in a brute force key guessing attack). 
She may even consider an attack where small amounts of information accumulated over a very long period of time eventually reveals significant amounts of information (precisely how this is achieved is not considered here).
A low level tampering attack is therefore an important consideration.

As noted above, any tampering attack would take the form of introducing the tamper variable $T$, to the message text via $R'_i=R_i\oplus T_i$ before it is sent on to Bob. 
Bob's local variable therefore becomes \[Y'_i = Y_i \oplus T_i = R_i \oplus N^B_i \oplus T_i.\] 
Eve cannot know $N^B$ therefore it is necessarily the case that $I(N^B;T)=0$.
The act of tampering modifies the binary symmetric channel between Alice and Bob (figure \ref{ErrorPicture}) changing the error rate from $\epsilon$ to $\epsilon'$, \[\epsilon'=(\beta+\tau)+\alpha-2(\beta+\tau)\alpha=\epsilon+\Delta\epsilon\]
where $\Delta\epsilon = \tau(1-2\alpha)$. 
The result is an increase in the channel noise between Alice and Bob without changing the channel noise between Eve and Alice, assuming that Eve does not incorporate the tamper variable in her own data (if she does, her error actually exceeds Bob's and secrecy is maintained). 

\begin{figure}[!t]
\centering
\begin{picture}(100,160)(10,0)
\put(-60,0){\includegraphics[width=3.2in]{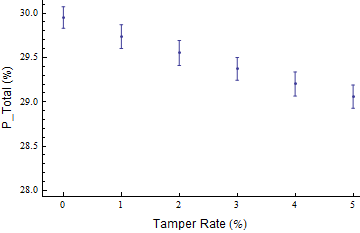}}
\end{picture}
\caption{This figure shows simulations results of the percentage of values retained ($P_{Total}$) at the first exchange when Alice and Bob introduce $16\%$ errors in their data, for an initial exchange of $5\times10^5$ bits. The error bars represent 1 standard deviation of error in the simulation results. Tampering is detectable at the $2-3\%$ level.}
\label{TamperPTotal}
\end{figure}

Figure \ref{TamperPTotal} shows the impact of tampering attacks on the value of $P_{Total}$ for attacks modifying $0-5\%$ of the data.
Hence, this particular tampering attack is detectable by examining $P_{Total}$ which reduces as  a function of $\tau$. 
Alice and Bob can monitor $P_{Total}$, with any significant drop indicating a tampering attack. 
A $2-5\%$ attack is clearly detectable and is outside a standard deviation of variation of the untampered value. 

This leaves a potential vulnerability to tampering attacks that introduce less than $2\%$ errors. 
The full protocol of key agreement now becomes important. 
It has already been shown that small amounts of tampering reduce Bob's mutual information with Alice and hence $P_{Total}$ in the first exchange, however it also, simultaneously, reduces Eve's mutual information with Bob.

In their exchange, and in the presence of tampering, Alice and Bob will arrive at the key source $S_T$ and $S'_T$ respectively which will have been reduced in size by Eve's tampering.
However, as key agreement is driven by Bob and not Alice, the final value of $I(S_T;S_T'')$ must therefore also be reduced, precisely because Eve has distanced herself from Bob. 
It remains the case that Eve's best attack is to do nothing.

Finally in completing the protocol, Alice and Bob, unaware of the tampering, can always use a Hash code derived without knowing the details of the attack and assuming a perfect Eve.
Figure \ref{TamperFinalKeyRate} shows the final key agreed by Alice and Bob when Eve has tampered with the data.
The key rate stays broadly flat with $I(\S_n;S_n'')\approx 0$. 
This is true for small levels of tampering and the situation changes if the error is substantially above $5\%$, hence continual monitoring of $P_{Total}$ is essential.
Alice should therefore also keep track of the information that Eve would have at each level of the protocol.

\begin{figure}[!t]
\centering
\begin{picture}(100,160)(10,0)
\put(-60,0){\includegraphics[width=3.2in]{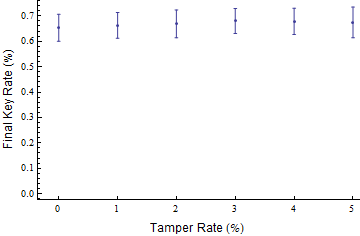}}
\end{picture}
\caption{This figure shows simulations the final secret key after four rounds of data matching, for an initial exchange of $5\times10^5$ bits. The final key rate is given by $(1-I(S_n;S''_n))P_{Final}$, where $P_{Final}$ is the final matching key agreed between Alice and Bob. The error bars represent 1 standard deviation statistical errors in the simulation results. In this simulation Alice and Bob introduce $16\%$ errors in their data.}
\label{TamperFinalKeyRate}
\end{figure}



\section{Conclusion}
This paper demonstrates the exchange of secret keys using an arbitrary communication channel with no {\it a priori} assumptions. 
This could therefore include channels such as the TCP/IP stack, or error corrected radio communication stacks. 
In this model Alice and Bob (the legitimate communicators) add noise to the data to simulate the action of a noisy channel. 
It has also been shown that a communication protocol exists that can be used to distill a secret key.

The nature of the noise is however key to the secrecy and therefore it must be unpredictable. 
Hence, the essential element is that the noise should be supplied by a truly random source. 
In practice this means that it should probably be sourced from quantum effects as these are both unpredictable and cannot, typically, be externally manipulated.

A feedback mechanism was also demonstrated where a new data set was constructed using a virtual error corrected channel exchanging data pairs.
This was shown to be capable of enhancing the overlap between Alice and Bob while reducing the noise in Eve. 
A potential attack has been considered where Eve attempts to degrade her data. This attack has been shown to lead to increased secrecy. 

The final step of extracting a key uses a two-way feedback protocol.
Such a protocol has not been considered in any detail in the current paper, however there are many examples in the literature that can be followed (see for example \cite{Silberhorn2002,VanAssche2004,Bennett1995,VanDijk1998,Brassard1994}).  

\section*{Acknowledgment}
Thanks go to W. Munro, T. Spiller and M. Everitt for useful feedback in assembling this paper.

\bibliographystyle{IEEEtran}

\bibliography{CryptoBib}

%



\end{document}